\documentclass[onecolumn,journal]{IEEEtran}
\ifCLASSINFOpdf
\else
\fi

\usepackage[T1]{fontenc}
\usepackage[latin9]{inputenc}
\usepackage{color}
\usepackage{array}
\usepackage{float}
\usepackage{mathrsfs}
\usepackage{mathtools}
\usepackage{amsthm}
\usepackage{amstext}
\usepackage{amssymb}
\usepackage{stmaryrd}
\usepackage{graphicx}
\usepackage{pgfplots}

\makeatletter

\theoremstyle{plain}
\newtheorem{theorem}{\protect\theoremname}
\theoremstyle{plain}

\theoremstyle{plain}

\theoremstyle{plain}
\newtheorem{lemma}[theorem]{\protect\lemmaname}
\theoremstyle{definition}
\newtheorem{example}[theorem]{\protect\examplename}
\theoremstyle{definition}

\theoremstyle{definition}
\newtheorem{remark}[theorem]{\protect\remarkname}

\AtBeginDocument{
	
}

\makeatother

  \providecommand{\corollaryname}{Corollary}
  \providecommand{\examplename}{Example}
  \providecommand{\lemmaname}{Lemma}
  \providecommand{\propositionname}{Proposition}
  \providecommand{\theoremname}{Theorem}
  \providecommand{\definitionname}{Definition}
  \providecommand{\remarkname}{Remark}

\begin{document}

%
\title{Complete Weight Enumerator of a Family of Linear Codes from Cyclotomy}
%
%
%

\author{Shudi Yang    and
        Zheng-An Yao and
        Chang-An Zhao 
	
\thanks{S.D. Yang is with the Department of Mathematics,
Sun Yat-sen University, Guangzhou 510275,
              and School of Mathematical
Sciences, Qufu Normal University, Shandong 273165, P.R.China.\protect\\

Z.-A. Yao and C.-A Zhao are with the
               Department of Mathematics,
Sun Yat-sen University, Guangzhou 510275, P.R.China.\protect\\
	\protect\\
	E-mail: yangshd3@mail2.sysu.edu.cn,~{mcsyao@mail.sysu.edu.cn}}
\thanks{Manuscript received *********; revised ********.}
}

%
%

\markboth{Journal of \LaTeX\ Class Files,~Vol.~13, No.~9, September~2014}%
{Shell \MakeLowercase{\textit{et al.}}: Bare Demo of IEEEtran.cls for Journals}
%



\maketitle

\begin{abstract}
Linear codes have been an interesting topic in both theory and practice for many years.
In this paper, for a prime $p$, we determine the explicit complete weight enumerators of a family of linear codes over $\mathbb{F}_p$ with defining set related to cyclotomy. These codes may have applications in cryptography and secret sharing schemes.
\end{abstract}


%
\IEEEpeerreviewmaketitle

\section{Introduction}\label{sec:intro}

 Throughout this paper, let $p$ be a prime and $r=p^m$ for some positive integer $m$. Denote by $\mathbb{F}_p$ a
finite field with $p$ elements. An $[n, \kappa, \delta]$ linear code
$C$ over $\mathbb{F}_p$ is a $\kappa$-dimensional subspace of
$\mathbb{F}_p^n$ with minimum distance $\delta$~\cite{macwilliams1977theory,ding2014differencesets}.

The complete weight enumerator~\cite{macwilliams1977theory,Blake1991} of a code $C$ over $\mathbb{F}_p$ enumerates the codewords according to the number of symbols of each kind contained in each codeword. Denote the
field elements by $\mathbb{F}_p=\{w_0,w_1,\cdots,w_{p-1}\}$, where $w_0=0$.
Also let $\mathbb{F}_p^*$ denote $\mathbb{F}_p\backslash\{0\}$.
For a codeword $\mathsf{c}=(c_0,c_1,\cdots,c_{n-1})\in \mathbb{F}_p^n$, let $w[\mathsf{c}]$ be the
complete weight enumerator of $\mathsf{c}$ defined as
$$w[\mathsf{c}]=w_0^{k_0}w_1^{k_1}\cdots w_{p-1}^{k_{p-1}},$$
where $k_j$ is the number of components of $\mathsf{c}$ equal to $w_j$, $\sum_{j=0}^{p-1}k_j=n$.
The complete weight enumerator of the code $C$ is then
$$\mathrm{CWE}(C)=\sum_{\mathsf{c}\in C}w[\mathsf{c}].$$

The weight distribution of a linear code has been studied extensively for decades and we
refer the reader to \cite{zheng2015weightseveral,yang2015weightenu,ding2011,ding2013hamming,dinh2015recent,feng2008weight,li2014weight,luo2008weight,sharma2012weight,vega2012weight,wang2012weight,yuan2006weight,zheng2015weight} and references therein for an overview of the related researches. It is not difficult to see that the complete weight enumerators are just the (ordinary) weight
enumerators for binary linear codes. For nonbinary linear codes, the weight enumerators can be obtained from their complete weight enumerators.

 The information of the complete weight enumerator of a linear code is of vital use in theories and practical applications. For instance, Blake and Kith investigated the complete weight enumerator of Reed-Solomon codes and showed that they could be helpful in soft decision decoding~\cite{Blake1991,kith1989complete}. In~\cite{helleseth2006}, the study of the monomial and quadratic bent functions was related to the complete weight enumerators of linear codes. It was illustrated by Ding $et~al.$~\cite{ding2007generic,Ding2005auth} that the complete weight enumerator can be applied to calculate the deception probabilities of certain authentication codes. In~\cite{chu2006constantco,ding2008optimal,ding2006construction}, the authors studied the complete weight enumerators of some constant
composition codes and presented some families of optimal constant composition codes.

However, it is generally an extremely difficult problem to evaluate the complete
weight enumerator of linear codes and there are few information on this topic in literature besides the above mentioned~\cite{Blake1991,kith1989complete,chu2006constantco,ding2008optimal,ding2006construction}.
Kuzmin and Nechaev considered the
generalized Kerdock code and related linear codes over Galois rings and determined their complete weight enumerators in~\cite{kuzmin1999complete} and \cite{kuzmin2001complete}. Very recently, the complete weight enumerators of linear codes over finite fields were studied in~\cite{li2015complete,li2015linear,BaeLi2015complete,yang2015linear,yang2015cyclic}.

In~\cite{ding2014differencesets,ding2007cyclotomic,dingkelan2014binary,ding2015twodesign,ding2015twothree}, Ding $et~al.$ proposed a generic method to construct linear codes with a few nonzero weights by employing trace function. We introduce this construction below.

Let $D=\{d_1,d_2,\cdots,d_z\}\subseteq \mathbb{F}_{r}$ for a positive integer $z$. Denote by $\mathrm{Tr}$ the trace function from $\mathbb{F}_{r}$ to $\mathbb{F}_{p}$. A linear code of length $z$ over $\mathbb{F}_{p}$ is defined by
\begin{equation}\label{def:CD}
    C_{D}=\{(\mathrm{Tr}(xd))_{d\in D}:
       x\in \mathbb{F}_{r}\},
\end{equation} and $D$ is called the defining set of this code $C_{D}$.

The authors in~\cite{ding2014differencesets,ding2007cyclotomic,dingkelan2014binary,ding2015twodesign,ding2015twothree} presented such linear codes and investigated their weight enumerators for some
well chosen defining sets. Along this research line, the authors of~\cite{li2015linear} and~\cite{yang2015linear} investigated
the complete weight enumerators of linear codes with defining sets for certain special cases.

Let $\alpha$ be a fixed primitive element of $\mathbb{F}_r$ and $r-1=nN$ for positive integers $n>1$ and $N>1$. We always assume that $N|\frac{r-1}{p-1}$. The cyclotomic classes of order $N$ are defined by $C_i^{(N,r)}=\alpha
^i\left<\alpha^N\right>$ for $i=0,1,\cdots,N-1$, where
$\langle\alpha^N\rangle$ denotes the subgroup of $\mathbb{F}_r^*$
generated by $ \alpha^N$.

In~\cite{ding2007cyclotomic}, Ding and Niederreiter constructed two classes of cyclotomic linear codes $C_{\bar{D}}$ over $\mathbb{F}_p$ of order $3$ with defining sets $\bar{D}=D_0$ and $\bar{D}=D_0\bigcup D_1$, where  $D_j=\{\alpha^j\alpha^{3i}:i=0,1,\cdots,\frac{r-1}{3(p-1)}-1\}$ for $j=0,1$. Inspired by the original idea proposed in~\cite{ding2007cyclotomic}, we shall study the complete weight enumerator of $C_D$ with defining sets constructed from cyclotomy.

In this paper, the defining set $D$ of the code $C_D$ is given by
\begin{eqnarray}\label{def:D}
D=\bigcup_{i\in I}C_i^{(N,r)},
\end{eqnarray}
where $I\subset\{0,1,\cdots,N-1\}$ with $\#I=l>0$.

Therefore the code $C_D$ defined by \eqref{def:CD} with defining set $D$ of \eqref{def:D} is a linear code with length $nl$ and dimension at most $m$. Obviously, the different choices of $I$ lead to different codes $C_D$.

 We employ Gaussian periods to determine the complete weight enumerator of $C_D$. A general formula is given for the defining set $D$ of \eqref{def:D}. Moreover, as applications of this formula, we explicitly present the complete weight enumerator of $C_D$ for the special cases of $N=3$ and $N=4$, respectively. In fact, the defining set of~\cite{ding2007cyclotomic} is a complete set of coset representatives of the factor group $C_0^{(3,r)}/{\mathbb{F}_p^*}$. We generalize it to the whole coset $C_i^{(3,r)}$ with $i=0,1,2$. As it turns out that, the codes $C_D$ for $N=3$ and $N=4$ are linear codes with few weights. More precisely, they have nonzero weights not more than $four$ and thus will have many applications in cryptography~\cite{YD2006} and secret sharing schemes~\cite{carlet2005linear}. They can be employed to construct constant composition codes~\cite{chu2006constantco,Luo2011constantco} which have important use in communications engineering~\cite{Milenkovic2006designDNA}. We also give some examples to illustrate our results, which shows that some of these codes are optimal or have the best parameters due to Ding's tables~\cite{ding2014differencesets}.


It should be remarked that, when $D=C_0^{(N,r)}$, our result is as same as that of~\cite{li2015complete}. We give the result by Gaussian periods and the authors of~\cite{li2015complete} gave the result by Gauss sums. They are equivalent though in different manner. In addition, when $N=3$, since $C_0^{(3,r)}=\{xy:x\in\mathbb{F}_{p}^* ~~\mathrm{and}~~y\in D_0\}=(\mathbb{F}_{p}^*)D_0$,  the weight enumerator of $C_{\bar{D}}$ in~\cite{ding2007cyclotomic} can be obtained from that of $C_D$. Thus, we extend the results in~\cite{ding2007cyclotomic} to some extent. These will be shown in details in the consequent sections.

The remainder of this paper is organized as follows. Section~\ref{sec:mathtool}
recalls some definitions and results about cyclotomic
classes and Gaussian periods which will be useful in the sequel. Section~\ref{sec:main} presents the complete weight enumerator of the code $C_D$ with defining set $D$, including a general strategy and the special case of $N=3$ and $N=4$. Section~\ref{sec:main}
concludes this paper.

\section{Mathematical foundation of Cyclotomy in $\mathbb{F}_r$}\label{sec:mathtool}

In this section, we introduce some necessary mathematical foundation which will be of use in the sequel.

Recall that $p$ is a prime, $r=p^{m}$ and $r-1=nN$ for two positive integers $n>1$ and $N>1$. Let $\alpha$ be a fixed primitive element of $\mathbb{F}_r$. Define $C_i^{(N,r)}=\alpha
^i\left<\alpha^N\right>$ for $i=0,1,\cdots,N-1$, where
$\langle\alpha^N\rangle$ denotes the subgroup of $\mathbb{F}_r^*$
generated by $ \alpha^N$. The cosets $C_i^{(N,r)}$ are called the \emph{cyclotomic classes} of order $N$. It is easily seen
that $\# C_i^{(N,r)}=\frac{r-1}{N}=n$ and
$C_i^{(N,r)}=C_{i\pmod N}^{(N,r)}$. Moreover, $C_0^{(N,r)},~C_1^{(N,r)},\cdots,C_{N-1}^{(N,r)}$ and $\{0\}$ form a partition of $\mathbb{F}_r$.

The \emph{Gaussian
periods} are defined by
 $$\eta_{i}^{(N,r)}=\sum_{x\in C_i^{(N,r)}}\zeta_p^{\mathrm{Tr}(x)},$$
 where $\zeta_p=e^{2\pi\sqrt{-1}/p}$ and $\mathrm{Tr}$ is the trace function from $\mathbb{F}_r$
to $\mathbb{F}_p$. Clearly $\eta_{i}^{(N,r)}=\eta_{i\pmod N}^{(N,r)}$ if $i\geqslant N$.

 Generally speaking, it is very hard to compute the values of Gaussian periods. They are known only in a few cases and they can be obtained from \emph{period polynomials} $\Psi_{(N,r)}(X)$ which are defined by
%
%
%
%

$$\Psi_{(N,r)}(X)=\prod_{i=0}^{N-1}(X-\eta_{i}^{(N,r)}).$$

It was shown that $\Psi_{(N,r)}(X)$ is a polynomial with integer coefficients~\cite{myerson1981period}. The following two lemmas, which were cited from~\cite{myerson1981period}, will be of use in the sequel.
%

 \begin{lemma}\emph{\cite{myerson1981period}}\label{periodpolyfac fac N=3}
Let $N=3$ and $r=p^{m}$. We have the following results on the factorization of $\Psi_{(3,r)}(X)$.\\
(a) If $p\equiv2 \pmod3$, then $m\equiv0 \pmod2$, and
\begin{eqnarray*}
\Psi_{(3,r)}(X)=\left\{\begin{array}{lll}
  3^{-3}(3X+1+2\sqrt{r})(3X+1-\sqrt{r})^2  &&\mathrm{if}~~m\equiv0 \pmod4,\\
  3^{-3}(3X+1-2\sqrt{r})(3X+1+\sqrt{r})^2  &&\mathrm{if}~~m\equiv2 \pmod4.\\
  \end{array} \right.
\end{eqnarray*}
(b) If $p\equiv1\pmod3$, and $m\equiv0 \pmod3$, then
\begin{eqnarray*}\Psi_{(3,r)}(X)=3^{-3}(3X+1-s_1 r^{\frac{1}{3}})
(3X+1+\frac{1}{2}(s_1+9t_1) r^{\frac{1}{3}})(3X+1+\frac{1}{2}(s_1-9t_1) r^{\frac{1}{3}}).
\end{eqnarray*}
where $s_1$ and $t_1$ are given by $4r^{\frac{1}{3}}=s_1^2+27t_1^2$, $s_1\equiv1\pmod3$ and $\mathrm{gcd}(s_1,p)=1$.
 \end{lemma}


  \begin{lemma}\emph{\cite{myerson1981period}}\label{periodpoly fac N=4}
Let $N=4$ and $r=p^{m}$. We have the following results on the factorization of $\Psi_{(4,r)}(X)$.\\
(a) If $p\equiv3 \pmod4$, then $m\equiv0 \pmod2$, and
\begin{eqnarray*}
\Psi_{(4,r)}(X)=\left\{\begin{array}{lll}
  4^{-4}(4X+1+3\sqrt{r})(4X+1-\sqrt{r})^3  &&\mathrm{if}~~m\equiv0 \pmod4,\\
  4^{-4}(4X+1-3\sqrt{r})(4X+1+\sqrt{r})^3  &&\mathrm{if}~~m\equiv2 \pmod4.\\
  \end{array} \right.
\end{eqnarray*}
(b) If $p\equiv1\pmod4$, and $m\equiv0 \pmod4$, then
\begin{eqnarray*}\Psi_{(4,r)}(X)&=&4^{-4}\left((4X+1)+\sqrt{r}+2r^{\frac{1}{4}}u_1\right)\left((4X+1)+\sqrt{r}-2r^{\frac{1}{4}}u_1\right)\\
&&\times\left((4X+1)-\sqrt{r}+4r^{\frac{1}{4}}v_1\right)
\left((4X+1)-\sqrt{r}-4r^{\frac{1}{4}}v_1\right),
\end{eqnarray*}
where $u_1$ and $v_1$ are given by $r^{\frac{1}{2}}=u_1^2+4v_1^2$, $u_1\equiv1\pmod4$ and $\mathrm{gcd}(u_1,p)=1$.
 \end{lemma}

\section{Main results}\label{sec:main}

We maintain al notations from the previous sections, and we want now to determine the complete weight enumerator of the codes $C_D$ defined by \eqref{def:CD} with defining set $D$ of \eqref{def:D}. These codes may have different property with different defining set $D$. Thus, we focus on a general case and two special cases.

\subsection{A general strategy}\label{subsec:general}

Recall that
\begin{equation*}
    C_{D}=\{(\mathrm{Tr}(xd))_{d\in D}:
       x\in \mathbb{F}_{r}\},
\end{equation*}
where $D=\bigcup_{i\in I}C_i^{(N,r)}$ with $I\subset\{0,1,\cdots,N-1\}$ and $\#I=l>0$.

Clearly $x=0$ gives the zero codeword which contributes $w_0^{nl}$ to the complete
weight enumerator. Thus, we only need to focus on $x\in\mathbb{F}_r^*$.

Let
\begin{equation*}
    N_k(\rho)=\#\{d\in D:\mathrm{Tr}(xd)=\rho\ \mathrm{and ~~}x\in C_k^{(N,r)}\},
\end{equation*}where $0\leqslant k \leqslant N-1$.

Note that $N|\frac{r-1}{p-1}$ leads to $\mathbb{F}_p^*\subseteq C_0^{(N,r)}$. Thus, we can deduce that
\begin{eqnarray}\label{eq:general N0}
   N_k(\rho)&=&\sum_{d\in D}
  \frac{1}{p}\sum_{y\in\mathbb{F}_{p}}\zeta_p^{y(\mathrm{Tr}(x d)-\rho)}\nonumber\\
  &=&\frac{n l }{p}+\frac{1}{p}\sum_{y\in\mathbb{F}_{p}^*}\zeta_p^{-y\rho}\sum_{d\in D}\zeta_p^{y\mathrm{Tr}(x d)}\nonumber\\
  &=&\frac{ n l}{p}+\frac{1}{p}\sum_{y\in\mathbb{F}_{p}^*}\zeta_p^{-y\rho}
      \left(\sum_{i\in I}\sum_{d\in C_i^{(N,r)}}\zeta_p^{\mathrm{Tr}(y x d)}\right)\nonumber\\
  &=&\frac{n l}{p}+\frac{1}{p}\sum_{y\in\mathbb{F}_{p}^*}\zeta_p^{-y\rho}
\sum_{i\in I}\eta_{k+i}^{(N,r)}\nonumber\\
  &=&\left\{\begin{array}{lll}
  \frac{n l}{p}+\frac{p-1}{p}\sum_{i\in I}\eta_{k+i}^{(N,r)}&&\mathrm{if}~~\rho=0,\\
  \frac{n l}{p}-\frac{1}{p}\sum_{i\in I}\eta_{k+i}^{(N,r)}&&\mathrm{otherwise},\\\end{array} \right.
\end{eqnarray}
where the fourth equality holds since $\mathbb{F}_p^*\subseteq C_0^{(N,r)}$.

Therefore, it follows from Equation \eqref{eq:general N0} that
\begin{eqnarray}\label{eq:generalcwe}
CWE(C_D)=w_0^{nl}+n\sum_{k=0}^{N-1}
w_0^{\frac{nl}{p}+\frac{p-1}{p}\sum_{i\in I}\eta_{k+i}^{(N,r)}}
\prod_{\rho=1}^{p-1}w_{\rho}^{\frac{nl}{p}-\frac{1}{p}\sum_{i\in I}\eta_{k+i}^{(N,r)}},
\end{eqnarray}
where $n=\# C_k^{(N,r)}=\frac{r-1}{N}$.

\subsection{Linear codes from cyclotomy of order N=3}\label{subsec:N=3}

In this subsection, we will give the complete weight enumerator of the code $C_D$ with defining set $D$ constructed from
cyclotomy of order $N=3$.

The results on Gaussian periods of order $N=3$ follow immediately from Lemma \ref{periodpolyfac fac N=3}.
\begin{lemma}\label{lem:gauss period of order 3}
Let $N=3$, $r=p^m$ and $N|\frac{r-1}{p-1}$. We have the following results on Gaussian periods.\\
(a) If $p\equiv1 \pmod3$, then $m\equiv0 \pmod3$ and
\begin{eqnarray*}
\eta_0^{(3,r)} &=& -\frac{1-s_1r^{\frac{1}{3}}}{3}, \\
\eta_1^{(3,r)} &=& -\frac{2+(s_1+9t_1)r^{\frac{1}{3}}}{6}, \\
\eta_2^{(3,r)} &=& -\frac{2+(s_1-9t_1)r^{\frac{1}{3}}}{6}, \\
\end{eqnarray*}
where $s_1$ and $t_1$ are defined in Lemma \ref{periodpolyfac fac N=3}.\\
(b) If $p\equiv2 \pmod3$, then $m\equiv0 \pmod2$.\\
For the case of $m\equiv0 \pmod4$, we have
\begin{eqnarray*}
\eta_0^{(3,r)} &=& -\frac{1+2\sqrt{r}}{3}, \\
\eta_1^{(3,r)} &=& \eta_2^{(3,r)} = -\frac{1-\sqrt{r}}{3}. \\
\end{eqnarray*}
For the case of $m\equiv2 \pmod4$, we have
\begin{eqnarray*}
\eta_0^{(3,r)} &=& -\frac{1-2\sqrt{r}}{3}, \\
\eta_1^{(3,r)} &=& \eta_2^{(3,r)} = -\frac{1+\sqrt{r}}{3}. \\
\end{eqnarray*}
\end{lemma}

The following gives the complete weight enumerator of the code $C_D$ with defining set $D$ constructed from
cyclotomy of order $N=3$.

\begin{theorem}\label{lem:cwe order 32}
Let $N=3$, $r=p^m$ and $N|\frac{r-1}{p-1}$. Let $\#I=2$ and $r-1=nN$. Then the code $C_D$ of \eqref{def:CD} is a $[2n,m]$ linear code over $\mathbb{F}_p$ and its complete weight enumerator is given as follows.\\
(a) If $p\equiv1 \pmod3$ and $m\equiv0 \pmod3$, then
\begin{eqnarray*}
CWE(C_D)&=&w_0^{2n}+nw_0^{\frac{2n}{p}-\frac{p-1}{6p}(4-(s_1-9t_1)r^{\frac{1}{3}})}
\prod_{\rho=1}^{p-1}w_{\rho}^{\frac{2n}{p}+\frac{1}{6p}(4-(s_1-9t_1)r^{\frac{1}{3}})}\\
&&+nw_0^{\frac{2n}{p}-\frac{p-1}{3p}(2+s_1r^{\frac{1}{3}})}
\prod_{\rho=1}^{p-1}w_{\rho}^{\frac{2n}{p}+\frac{1}{3p}(2+s_1r^{\frac{1}{3}})}\\
&&+nw_0^{\frac{2n}{p}-\frac{p-1}{6p}(4-(s_1+9t_1)r^{\frac{1}{3}})}
\prod_{\rho=1}^{p-1}w_{\rho}^{\frac{2n}{p}+\frac{1}{6p}(4-(s_1+9t_1)r^{\frac{1}{3}})}.\\
\end{eqnarray*}
(b) If $p\equiv2 \pmod3$ and $m\equiv0 \pmod4$, then
\begin{eqnarray*}
CWE(C_D)&=&w_0^{2n}+2nw_0^{\frac{2n}{p}-\frac{p-1}{3p}(2+\sqrt{r})}
\prod_{\rho=1}^{p-1}w_{\rho}^{\frac{2n}{p}+\frac{1}{3p}(2+\sqrt{r})}\\
&&+nw_0^{\frac{2n}{p}-\frac{2(p-1)}{3p}(1-\sqrt{r})}
\prod_{\rho=1}^{p-1}w_{\rho}^{\frac{2n}{p}+\frac{2}{3p}(1-\sqrt{r})}.\\
\end{eqnarray*}
(c) If $p\equiv2 \pmod3$ and $m\equiv2 \pmod4$, then
\begin{eqnarray*}
CWE(C_D)&=&w_0^{2n}+2nw_0^{\frac{2n}{p}-\frac{p-1}{3p}(2-\sqrt{r})}
\prod_{\rho=1}^{p-1}w_{\rho}^{\frac{2n}{p}+\frac{1}{3p}(2-\sqrt{r})}\\
&&+nw_0^{\frac{2n}{p}-\frac{2(p-1)}{3p}(1+\sqrt{r})}
\prod_{\rho=1}^{p-1}w_{\rho}^{\frac{2n}{p}+\frac{2}{3p}(1+\sqrt{r})}.\\
\end{eqnarray*}
\end{theorem}
\begin{proof}
It follows from Lemma \ref{lem:gauss period of order 3} that the following assertions hold.\\
(a) If $p\equiv1 \pmod3$ and $m\equiv0 \pmod3$, then
\begin{eqnarray*}
\eta_0^{(3,r)}+\eta_1^{(3,r)} &=& -\frac{4-(s_1-9t_1)r^{\frac{1}{3}}}{6}, \\
\eta_1^{(3,r)}+\eta_2^{(3,r)} &=& -\frac{2+s_1r^{\frac{1}{3}}}{3}, \\
\eta_2^{(3,r)}+\eta_0^{(3,r)} &=& -\frac{4-(s_1+9t_1)r^{\frac{1}{3}}}{6}, \\
\end{eqnarray*}
where $s_1$ and $t_1$ are defined in Lemma \ref{periodpolyfac fac N=3}.\\
(b) If $p\equiv2 \pmod3$ and $m\equiv0 \pmod4$, then
\begin{eqnarray*}
\eta_0^{(3,r)}+\eta_1^{(3,r)} &=& -\frac{2+\sqrt{r}}{3}=\eta_2^{(3,r)}+\eta_0^{(3,r)}, \\
\eta_1^{(3,r)}+\eta_2^{(3,r)} &=& -\frac{2(1-\sqrt{r})}{3}.\\
\end{eqnarray*}
(c) If $p\equiv2 \pmod3$ and $m\equiv2 \pmod4$, then
\begin{eqnarray*}
\eta_0^{(3,r)}+\eta_1^{(3,r)} &=& -\frac{2-\sqrt{r}}{3}=\eta_2^{(3,r)}+\eta_0^{(3,r)}, \\
\eta_1^{(3,r)}+\eta_2^{(3,r)} &=& -\frac{2(1+\sqrt{r})}{3}.\\
\end{eqnarray*}
The desired conclusions then follow from Equation \eqref{eq:generalcwe}.
\end{proof}

When $\#I=1$, the result is straightforward from Lemma \ref{lem:gauss period of order 3} and Equation \eqref{eq:generalcwe}.
\begin{theorem}\label{lem:cwe order 31}
Let $N=3$, $r=p^m$ and $N|\frac{r-1}{p-1}$. Let $\#I=1$ and $r-1=nN$. Then the code $C_D$ of \eqref{def:CD} is an $[n,m]$ cyclic code over $\mathbb{F}_p$ and its complete weight enumerator is given as follows.\\
(a) If $p\equiv1 \pmod3$ and $m\equiv0 \pmod3$, then
\begin{eqnarray*}
CWE(C_D)&=&w_0^n+nw_0^{\frac{n}{p}-\frac{p-1}{3p}(1-s_1r^{\frac{1}{3}})}
\prod_{\rho=1}^{p-1}w_{\rho}^{\frac{n}{p}+\frac{1}{3p}(1-s_1r^{\frac{1}{3}})}\\
&&+nw_0^{\frac{n}{p}-\frac{p-1}{6p}(2+(s_1+9t_1)r^{\frac{1}{3}})}
\prod_{\rho=1}^{p-1}w_{\rho}^{\frac{n}{p}+\frac{1}{6p}(2+(s_1+9t_1)r^{\frac{1}{3}})}\\
&&+nw_0^{\frac{n}{p}-\frac{p-1}{6p}(2+(s_1-9t_1)r^{\frac{1}{3}})}
\prod_{\rho=1}^{p-1}w_{\rho}^{\frac{n}{p}+\frac{1}{6p}(2+(s_1-9t_1)r^{\frac{1}{3}})},\\
\end{eqnarray*}
where $s_1$ and $t_1$ are defined in Lemma \ref{periodpolyfac fac N=3}.\\
(b) If $p\equiv2 \pmod3$ and $m\equiv0 \pmod4$, then
\begin{eqnarray*}
CWE(C_D)&=&w_0^n+2nw_0^{\frac{n}{p}-\frac{p-1}{3p}(1-\sqrt{r})}
\prod_{\rho=1}^{p-1}w_{\rho}^{\frac{n}{p}+\frac{1}{3p}(1-\sqrt{r})}\\
&&+nw_0^{\frac{n}{p}-\frac{p-1}{3p}(1+2\sqrt{r})}
\prod_{\rho=1}^{p-1}w_{\rho}^{\frac{n}{p}+\frac{1}{3p}(1+2\sqrt{r})}.\\
\end{eqnarray*}
(c) If $p\equiv2 \pmod3$ and $m\equiv2 \pmod4$, then
\begin{eqnarray*}
CWE(C_D)&=&w_0^n+2nw_0^{\frac{n}{p}-\frac{p-1}{3p}(1+\sqrt{r})}
\prod_{\rho=1}^{p-1}w_{\rho}^{\frac{n}{p}+\frac{1}{3p}(1+\sqrt{r})}\\
&&+nw_0^{\frac{n}{p}-\frac{p-1}{3p}(1-2\sqrt{r})}
\prod_{\rho=1}^{p-1}w_{\rho}^{\frac{n}{p}+\frac{1}{3p}(1-2\sqrt{r})}.\\
\end{eqnarray*}
\end{theorem}

\begin{remark}
It should be remarked that the authors in~\cite{ding2007cyclotomic} determined the weight enumerators of two classes of cyclotomic linear codes $C_{\bar{D}}$ over $\mathbb{F}_p$ of order $3$ with defining sets $\bar{D}=D_0$ and $\bar{D}=D_0\bigcup D_1$, where $D_j=\{\alpha^j\alpha^{3i}:i=0,1,\cdots,\frac{r-1}{3(p-1)}-1\}$ for $j=0,1$.

Note that $C_0^{(3,r)}=\{xy:x\in\mathbb{F}_{p}^* ~~\mathrm{and}~~y\in D_0\}=(\mathbb{F}_{p}^*)D_0$. By Theorems \ref{lem:cwe order 32} and \ref{lem:cwe order 31}, we can obtain the weight enumerator of $C_D$ for $D=\bigcup_{i\in I}C_i^{(3,r)}$ with $I\subset\{0,1,2\}$. Dividing each nonzero weight of $C_D$ with $p-1$ yields the weight enumerator of $C_{\bar{D}}$, which conforms to the results described in~\cite{ding2007cyclotomic}. Thus we generalize the results of~\cite{ding2007cyclotomic} to some extent.
\end{remark}

\begin{example}
Let $(p,m)=(2,6)$ and $N=3$. Then $r=64$ and $n=21$.
Suppose that $\alpha$ is a primitive element of $\mathbb{F}_{64}$.\\

(1) For the case of $D=C_1^{(3,r)}=\alpha\langle\alpha^3\rangle$.

By Theorem \ref{lem:cwe order 31}, the code $C_{D}$ of \eqref{def:CD} is a binary $[21,6,8]$ cyclic code and its complete weight enumerator is
\begin{eqnarray*}
CWE(C_{D})=w_0^{21}+21w_0^{13}w_1^{8}+42w_0^{9}w_1^{12},
\end{eqnarray*}
which is consistent with numerical computation by Magma.

This code is the best binary cyclic code of length $21$ and dimension $6$ according to the tables given by Ding~\cite{ding2014differencesets}.

(2) For the case of $D=C_0^{(3,r)}\bigcup C_1^{(3,r)}= \langle\alpha^3\rangle\bigcup \alpha\langle\alpha^3\rangle$.

By Theorem \ref{lem:cwe order 32}, the code $C_{D}$ of \eqref{def:CD} is a binary $[42,6,20]$ cyclic code which is optimal with respect to
the Griesmer bound, and its complete weight enumerator is
\begin{eqnarray*}
CWE(C_{D})=w_0^{42}+42w_0^{22}w_1^{20}+21w_0^{18}w_1^{24},
\end{eqnarray*}
which is consistent with numerical computation by Magma.\\

\end{example}

\subsection{Linear codes from cyclotomy of order N=4}\label{subsec:N=4}

In this subsection, we will present the complete weight enumerator of the code $C_D$ with defining set $D$ constructed from
cyclotomy of order $N=4$.

The results on Gaussian periods of order $N=4$ follow immediately from Lemma \ref{periodpoly fac N=4}.
\begin{lemma}\label{lem:gauss period of order 4}
Let $N=4$, $r=p^m$ and $N|\frac{r-1}{p-1}$. We have the following results on Gaussian periods.\\
(a) If $p\equiv1 \pmod4$, then $m\equiv0 \pmod4$ and
\begin{eqnarray*}
\eta_0^{(4,r)} &=& -\frac{1+\sqrt{r}+2r^{\frac{1}{4}}u_1}{4}, \\
\eta_1^{(4,r)} &=& -\frac{1-\sqrt{r}+4r^{\frac{1}{4}}v_1}{4}, \\
\eta_2^{(4,r)} &=& -\frac{1+\sqrt{r}-2r^{\frac{1}{4}}u_1}{4}, \\
\eta_3^{(4,r)} &=& -\frac{1-\sqrt{r}-4r^{\frac{1}{4}}v_1}{4}, \\
\end{eqnarray*}
where $u_1$ and $v_1$ are defined in Lemma \ref{periodpoly fac N=4}.\\
(b) If $p\equiv3 \pmod4$, then $m\equiv0 \pmod2$.\\
For the case of $m\equiv0 \pmod4$, we have
\begin{eqnarray*}
\eta_0^{(4,r)} &=& -\frac{1+3\sqrt{r}}{4}, \\
\eta_j^{(4,r)} &=& -\frac{1-\sqrt{r}}{4} ~~for~~ all~~ j\neq0. \\
\end{eqnarray*}
For the case of $m\equiv2 \pmod4$, we have
\begin{eqnarray*}
\eta_0^{(4,r)} &=& -\frac{1-3\sqrt{r}}{4}, \\
\eta_j^{(4,r)} &=& -\frac{1+\sqrt{r}}{4} ~~for~~ all~~ j\neq0. \\
\end{eqnarray*}
\end{lemma}

The following gives the complete weight enumerator of the code $C_D$ with defining set $D$ constructed from
cyclotomy of order $N=4$.

\begin{theorem}\label{lem:cwe order 43}
Let $N=4$, $r=p^m$ and $N|\frac{r-1}{p-1}$. Let $\#I=3$ and $r-1=nN$. Then the code $C_D$ of \eqref{def:CD} is a $[3n,m]$ linear code over $\mathbb{F}_p$ and its complete weight enumerator is given as follows.\\
(a) If $p\equiv1 \pmod4$ and $m\equiv0 \pmod4$, then
\begin{eqnarray*}
CWE(C_D)&=&w_0^{3n}+nw_0^{\frac{3n}{p}-\frac{p-1}{4p}(3-\sqrt{r}+2r^{\frac{1}{4}}u_1)}
\prod_{\rho=1}^{p-1}w_{\rho}^{\frac{3n}{p}+\frac{1}{4p}(3-\sqrt{r}+2r^{\frac{1}{4}}u_1)}\\
&&+nw_0^{\frac{3n}{p}-\frac{p-1}{4p}(3-\sqrt{r}-2r^{\frac{1}{4}}u_1)}
\prod_{\rho=1}^{p-1}w_{\rho}^{\frac{3n}{p}+\frac{1}{4p}(3-\sqrt{r}-2r^{\frac{1}{4}}u_1)}\\
&&+nw_0^{\frac{3n}{p}-\frac{p-1}{4p}(3+\sqrt{r}+4r^{\frac{1}{4}}v_1)}
\prod_{\rho=1}^{p-1}w_{\rho}^{\frac{3n}{p}+\frac{1}{4p}(3+\sqrt{r}+4r^{\frac{1}{4}}v_1)}\\
&&+nw_0^{\frac{3n}{p}-\frac{p-1}{4p}(3+\sqrt{r}-4r^{\frac{1}{4}}v_1)}
\prod_{\rho=1}^{p-1}w_{\rho}^{\frac{3n}{p}+\frac{1}{4p}(3+\sqrt{r}-4r^{\frac{1}{4}}v_1)},\\
\end{eqnarray*}
where $u_1$ and $v_1$ are defined in Lemma \ref{periodpoly fac N=4}.\\
(b) If $p\equiv3 \pmod4$ and $m\equiv0 \pmod4$, then
\begin{eqnarray*}
CWE(C_D)&=&w_0^{3n}+3nw_0^{\frac{3n}{p}-\frac{p-1}{4p}(3+\sqrt{r})}
\prod_{\rho=1}^{p-1}w_{\rho}^{\frac{3n}{p}+\frac{1}{4p}(3+\sqrt{r})}\\
&&+nw_0^{\frac{3n}{p}-\frac{3(p-1)}{4p}(1-\sqrt{r})}
\prod_{\rho=1}^{p-1}w_{\rho}^{\frac{3n}{p}+\frac{3}{4p}(1-\sqrt{r})}.\\
\end{eqnarray*}
(c) If $p\equiv3 \pmod4$ and $m\equiv2 \pmod4$, then
\begin{eqnarray*}
CWE(C_D)&=&w_0^{3n}+3nw_0^{\frac{3n}{p}-\frac{p-1}{4p}(3-\sqrt{r})}
\prod_{\rho=1}^{p-1}w_{\rho}^{\frac{3n}{p}+\frac{1}{4p}(3-\sqrt{r})}\\
&&+nw_0^{\frac{3n}{p}-\frac{3(p-1)}{4p}(1+\sqrt{r})}
\prod_{\rho=1}^{p-1}w_{\rho}^{\frac{3n}{p}+\frac{3}{4p}(1+\sqrt{r})}.\\
\end{eqnarray*}
\end{theorem}

\begin{proof}
From Lemma \ref{lem:gauss period of order 4}, we can deduce the following assertions.\\
(a) If $p\equiv1 \pmod4$ and $m\equiv0 \pmod4$, then
\begin{eqnarray*}
\eta_0^{(4,r)}+\eta_1^{(4,r)}+\eta_2^{(4,r)} &=&
            -\frac{3+\sqrt{r}+4r^{\frac{1}{4}}v_1}{4}, \\
\eta_1^{(4,r)}+\eta_2^{(4,r)}+\eta_3^{(4,r)} &=&
            -\frac{3-\sqrt{r}-2r^{\frac{1}{4}}u_1}{4}, \\
\eta_2^{(4,r)}+\eta_3^{(4,r)}+\eta_0^{(4,r)} &=&
            -\frac{3+\sqrt{r}-4r^{\frac{1}{4}}v_1}{4}, \\
\eta_3^{(4,r)}+\eta_0^{(4,r)}+\eta_1^{(4,r)} &=&
            -\frac{3-\sqrt{r}+2r^{\frac{1}{4}}u_1}{4}, \\
\end{eqnarray*}
where $u_1$ and $v_1$ are defined in Lemma \ref{periodpoly fac N=4}.\\
(b) If $p\equiv3 \pmod4$ and $m\equiv0 \pmod4$, then
\begin{eqnarray*}
\eta_0^{(4,r)}+\eta_1^{(4,r)}+\eta_2^{(4,r)}
&=&\eta_2^{(4,r)}+\eta_3^{(4,r)}+\eta_0^{(4,r)}\\
&=&\eta_3^{(4,r)}+\eta_0^{(4,r)}+\eta_1^{(4,r)}\\
&=&  -\frac{3+\sqrt{r}}{4}, \\
\eta_1^{(4,r)}+\eta_2^{(4,r)}+\eta_3^{(4,r)} &=& -\frac{3(1-\sqrt{r})}{4}. \\
\end{eqnarray*}
(c) If $p\equiv3 \pmod4$ and $m\equiv2 \pmod4$, then
\begin{eqnarray*}
\eta_0^{(4,r)}+\eta_1^{(4,r)}+\eta_2^{(4,r)}
&=&\eta_2^{(4,r)}+\eta_3^{(4,r)}+\eta_0^{(4,r)}\\
&=&\eta_3^{(4,r)}+\eta_0^{(4,r)}+\eta_1^{(4,r)}\\
&=&  -\frac{3-\sqrt{r}}{4}, \\
\eta_1^{(4,r)}+\eta_2^{(4,r)}+\eta_3^{(4,r)} &=& -\frac{3(1+\sqrt{r})}{4}. \\
\end{eqnarray*}
The desired conclusions then follow from Equation \eqref{eq:generalcwe}.
\end{proof}

\begin{theorem}\label{lem:cwe order 42}
Let $N=4$, $r=p^m$ and $N|\frac{r-1}{p-1}$. Let $\#I=2$ and $r-1=nN$. Then the code $C_D$ of \eqref{def:CD} is a $[2n,m]$ linear code over $\mathbb{F}_p$ and its complete weight enumerator is given as follows.\\
(a) If $p\equiv1 \pmod4$ and $m\equiv0 \pmod4$, then
\begin{eqnarray*}
CWE(C_D)&=&w_0^{2n}+2nw_0^{\frac{2n}{p}-\frac{p-1}{2p}(1+\sqrt{r})}
\prod_{\rho=1}^{p-1}w_{\rho}^{\frac{2n}{p}+\frac{1}{2p}(1+\sqrt{r})}\\
&&+2nw_0^{\frac{2n}{p}-\frac{p-1}{2p}(1-\sqrt{r})}
\prod_{\rho=1}^{p-1}w_{\rho}^{\frac{2n}{p}+\frac{1}{2p}(1-\sqrt{r})},\\
\end{eqnarray*}
for the case of $I=\{0,2\}$ and $I=\{1,3\}$, and for other cases,
\begin{eqnarray*}
CWE(C_D)&=&w_0^{2n}+nw_0^{\frac{2n}{p}-\frac{p-1}{2p}(1+r^{\frac{1}{4}}(u_1+2v_1))}
\prod_{\rho=1}^{p-1}w_{\rho}^{\frac{2n}{p}+\frac{1}{2p}(1+r^{\frac{1}{4}}(u_1+2v_1))}\\
&&+nw_0^{\frac{2n}{p}-\frac{p-1}{2p}(1-r^{\frac{1}{4}}(u_1+2v_1))}
\prod_{\rho=1}^{p-1}w_{\rho}^{\frac{2n}{p}+\frac{1}{2p}(1-r^{\frac{1}{4}}(u_1+2v_1))}\\
&&+nw_0^{\frac{2n}{p}-\frac{p-1}{2p}(1+r^{\frac{1}{4}}(u_1-2v_1))}
\prod_{\rho=1}^{p-1}w_{\rho}^{\frac{2n}{p}+\frac{1}{2p}(1+r^{\frac{1}{4}}(u_1-2v_1))}\\
&&+nw_0^{\frac{2n}{p}-\frac{p-1}{2p}(1-r^{\frac{1}{4}}(u_1-2v_1))}
\prod_{\rho=1}^{p-1}w_{\rho}^{\frac{2n}{p}+\frac{1}{2p}(1-r^{\frac{1}{4}}(u_1-2v_1))},\\
\end{eqnarray*}
where $u_1$ and $v_1$ are defined in Lemma \ref{periodpoly fac N=4}.\\
(b) If $p\equiv3 \pmod4$ and $m\equiv0 \pmod2$, then
\begin{eqnarray*}
CWE(C_D)&=&w_0^{2n}+2nw_0^{\frac{2n}{p}-\frac{p-1}{2p}(1+\sqrt{r})}
\prod_{\rho=1}^{p-1}w_{\rho}^{\frac{2n}{p}+\frac{1}{2p}(1+\sqrt{r})}\\
&&+2nw_0^{\frac{2n}{p}-\frac{p-1}{2p}(1-\sqrt{r})}
\prod_{\rho=1}^{p-1}w_{\rho}^{\frac{2n}{p}+\frac{1}{2p}(1-\sqrt{r})}.\\
\end{eqnarray*}
\end{theorem}

\begin{proof}
From Lemma \ref{lem:gauss period of order 4}, we can deduce the following assertions.\\
(a) If $p\equiv1 \pmod4$ and $m\equiv0 \pmod4$.\\
For the case of $I=\{0,2\}$ and $I=\{1,3\}$, we have
\begin{eqnarray*}
\eta_0^{(4,r)}+\eta_2^{(4,r)} &= -\frac{1+\sqrt{r}}{2}&=\eta_2^{(4,r)}+\eta_0^{(4,r)}, \\
\eta_1^{(4,r)}+\eta_3^{(4,r)} &= -\frac{1-\sqrt{r}}{2}&=\eta_3^{(4,r)}+\eta_1^{(4,r)} . \\
\end{eqnarray*}
And, for other cases, we have
\begin{eqnarray*}
\eta_0^{(4,r)}+\eta_1^{(4,r)} &=& -\frac{1+r^{\frac{1}{4}}(u_1+2v_1)}{2}, \\
\eta_1^{(4,r)}+\eta_2^{(4,r)} &=& -\frac{1-r^{\frac{1}{4}}(u_1-2v_1)}{2}, \\
\eta_2^{(4,r)}+\eta_3^{(4,r)} &=& -\frac{1-r^{\frac{1}{4}}(u_1+2v_1)}{2}, \\
\eta_3^{(4,r)}+\eta_0^{(4,r)} &=& -\frac{1+r^{\frac{1}{4}}(u_1-2v_1)}{2}. \\
\end{eqnarray*}
where $u_1$ and $v_1$ are defined in Lemma \ref{periodpoly fac N=4}.\\
(b) If $p\equiv3 \pmod4$ and $m\equiv0 \pmod2$.\\
For the case of $m\equiv0 \pmod4$, we have
\begin{eqnarray*}
\eta_0^{(4,r)}+\eta_1^{(4,r)} &= -\frac{1+\sqrt{r}}{2} &= \eta_3^{(4,r)}+\eta_0^{(4,r)},\\
\eta_1^{(4,r)}+\eta_2^{(4,r)} &= -\frac{1-\sqrt{r}}{2} &= \eta_2^{(4,r)}+\eta_3^{(4,r)}.\\
\end{eqnarray*}
For the case of $m\equiv2 \pmod4$, we have
\begin{eqnarray*}
\eta_0^{(4,r)}+\eta_1^{(4,r)} &= -\frac{1-\sqrt{r}}{2} &= \eta_3^{(4,r)}+\eta_0^{(4,r)},\\
\eta_1^{(4,r)}+\eta_2^{(4,r)} &= -\frac{1+\sqrt{r}}{2} &= \eta_2^{(4,r)}+\eta_3^{(4,r)}.\\
\end{eqnarray*}
The desired conclusions then follow from Equation \eqref{eq:generalcwe}.
\end{proof}

When $\#I=1$, the result is straightforward from Lemma \ref{lem:gauss period of order 4} and Equation \eqref{eq:generalcwe}.

\begin{theorem}\label{lem:cwe order 41}
Let $N=4$, $r=p^m$ and $N|\frac{r-1}{p-1}$. Let $\#I=1$ and $r-1=nN$. Then the code $C_D$ of \eqref{def:CD} is an $[n,m]$ cyclic code over $\mathbb{F}_p$ and its complete weight enumerator is given as follows.\\
(a) If $p\equiv1 \pmod4$ and $m\equiv0 \pmod4$, then
\begin{eqnarray*}
CWE(C_D)&=&w_0^n+nw_0^{\frac{n}{p}-\frac{p-1}{4p}(1+\sqrt{r}+2r^{\frac{1}{4}}u_1)}
\prod_{\rho=1}^{p-1}w_{\rho}^{\frac{n}{p}+\frac{1}{4p}(1+\sqrt{r}+2r^{\frac{1}{4}}u_1)}\\
&&+nw_0^{\frac{n}{p}-\frac{p-1}{4p}(1+\sqrt{r}-2r^{\frac{1}{4}}u_1)}
\prod_{\rho=1}^{p-1}w_{\rho}^{\frac{n}{p}+\frac{1}{4p}(1+\sqrt{r}-2r^{\frac{1}{4}}u_1)}\\
&&+nw_0^{\frac{n}{p}-\frac{p-1}{4p}(1-\sqrt{r}+4r^{\frac{1}{4}}v_1)}
\prod_{\rho=1}^{p-1}w_{\rho}^{\frac{n}{p}+\frac{1}{4p}(1-\sqrt{r}+4r^{\frac{1}{4}}v_1)}\\
&&+nw_0^{\frac{n}{p}-\frac{p-1}{4p}(1-\sqrt{r}-4r^{\frac{1}{4}}v_1)}
\prod_{\rho=1}^{p-1}w_{\rho}^{\frac{n}{p}+\frac{1}{4p}(1-\sqrt{r}-4r^{\frac{1}{4}}v_1)},\\
\end{eqnarray*}
where $u_1$ and $v_1$ are defined in Lemma \ref{periodpoly fac N=4}.\\
(b) If $p\equiv3 \pmod4$ and $m\equiv0 \pmod4$, then
\begin{eqnarray*}
CWE(C_D)&=&w_0^n+3nw_0^{\frac{n}{p}-\frac{p-1}{4p}(1-\sqrt{r})}
\prod_{\rho=1}^{p-1}w_{\rho}^{\frac{n}{p}+\frac{1}{4p}(1-\sqrt{r})}\\
&&+nw_0^{\frac{n}{p}-\frac{p-1}{4p}(1+3\sqrt{r})}
\prod_{\rho=1}^{p-1}w_{\rho}^{\frac{n}{p}+\frac{1}{4p}(1+3\sqrt{r})}.\\
\end{eqnarray*}
(c) If $p\equiv3 \pmod4$ and $m\equiv2 \pmod4$, then
\begin{eqnarray*}
CWE(C_D)&=&w_0^n+3nw_0^{\frac{n}{p}-\frac{p-1}{4p}(1+\sqrt{r})}
\prod_{\rho=1}^{p-1}w_{\rho}^{\frac{n}{p}+\frac{1}{4p}(1+\sqrt{r})}\\
&&+nw_0^{\frac{n}{p}-\frac{p-1}{4p}(1-3\sqrt{r})}
\prod_{\rho=1}^{p-1}w_{\rho}^{\frac{n}{p}+\frac{1}{4p}(1-3\sqrt{r})}.\\
\end{eqnarray*}
\end{theorem}

\begin{remark}
 We remark that the general strategy of Equation \eqref{eq:general N0} is equivalent to the formula given in Theorem $3.1$ of~\cite{li2015complete} for the special case of $D=C_0^{(N,r)}$. In other words,
\begin{eqnarray}\label{eq:equal}
\eta_{k}^{(N,r)}=\frac{1}{N}\sum_{i=0}^{N-1}G(\bar{\tau}^i)\tau^i(a),
\end{eqnarray}
where $a\in C_k^{(N,r)}$, $\tau=\chi^n$ and $\bar{\tau}$ be the conjugate character of $\tau$. We shall show this in detail.

Recall that $\alpha$ is the primitive element of $\mathbb{F}_{r}$. Let $\eta_{k}^{(N,r)}=\sum_{x\in C_k^{(N,r)}}\psi(x)$, where $\psi(x)=\zeta_p^{\mathrm{Tr}(x)}$ is the canonical additive character over $\mathbb{F}_r$.

Define Gauss sum over $\mathbb{F}_r$ to be
$$G(\lambda)=\sum_{x\in\mathbb{F}_r^*}\lambda(x)\psi(x),$$
where $\lambda$ is a multiplicative character of $\mathbb{F}_r$.

It is known that the set $\widehat{\mathbb{F}}_r^*$ of all the multiplicative characters of $\mathbb{F}_r^*$ forms a group generated by $\chi$, i.e., $\widehat{\mathbb{F}}_r^*=\left<\chi\right>$, where $\chi$ is a multiplicative character of order $r-1$. Then Gauss sums can be regarded as the Fourier coefficients in the Fourier expansion of the restriction of $\psi$ to $\mathbb{F}_r^*$ in terms of the multiplicative characters of $\mathbb{F}_r$. That is
\begin{eqnarray}\label{eq:equal}
\psi(x)=\frac{1}{r-1}\sum_{\lambda\in\left<\chi\right>}G(\bar{\lambda})\lambda(x), ~~\mathrm{for}~~ x\in \mathbb{F}_r^*.
\end{eqnarray}

With above preparation, the left hand side of \eqref{eq:equal} is
\begin{eqnarray*}
 LHS&=&\sum_{x\in C_k^{(N,r)}}\psi(x)\\
&=& \sum_{x\in C_k^{(N,r)}}\frac{1}{r-1}\sum_{\lambda\in\left<\chi\right>}G(\bar{\lambda})\lambda(x) \\
&=& \frac{1}{r-1}\sum_{\lambda\in\left<\chi\right>}G(\bar{\lambda})\sum_{j=0}^{n-1}\lambda(\alpha^k\alpha^{Nj}) \\
&=& \frac{1}{r-1}\sum_{\lambda\in\left<\chi\right>}G(\bar{\lambda})\lambda(\alpha^k)\sum_{j=0}^{n-1}\lambda(\alpha^{Nj}) \\
&=& \frac{n}{r-1}\sum_{\lambda\in\left<\chi^n\right>}G(\bar{\lambda})\lambda(\alpha^k) \\
&=& \frac{1}{N}\sum_{i=0}^{N-1}G(\bar{\chi}^{ni})\chi^{ni}(\alpha^k), \\
\end{eqnarray*}
where the fifth equal sign holds since
\begin{eqnarray*}
\sum_{j=0}^{n-1}\lambda(\alpha^{Nj})=\left\{\begin{array}{lll}
  n  &&\mathrm{if}~~\lambda^N=1,\\
  0 &&otherwise.\\
  \end{array} \right.
\end{eqnarray*}

On the other hand the right hand side of \eqref{eq:equal} is
\begin{eqnarray*}
RHS&=&\frac{1}{N}\sum_{i=0}^{N-1}G(\bar{\tau}^i)\tau^i(a)\\
&=& \frac{1}{N}\sum_{i=0}^{N-1}G(\bar{\chi}^{ni})\chi^{ni}(\alpha^k\alpha^{Nj_0}) \\
&=& \frac{1}{N}\sum_{i=0}^{N-1}G(\bar{\chi}^{ni})\chi^{ni}(\alpha^k). \\
\end{eqnarray*}

Therefore, Equation \eqref{eq:equal} holds for all $x\neq0$.

The results of case $p\equiv2 \pmod3$ in Theorem \ref{lem:cwe order 31} and case $p\equiv3 \pmod4$ in Theorem \ref{lem:cwe order 41} consistent with the result of Theorem $3.5$ in~\cite{li2015complete} for semi-primitive case.

We also mention that the general strategy can be applied to the case of $N\in\{5,6,8,12\}$, though we did not list them here,
by ultilizing the Gaussian periods of order $N\in\{5,6,8,12\}$ since the period polynomial $\Psi_{(N,r)}(X)$ and its factorization were determined explicitly for $N=5$ in~\cite{hoshi2006explicit}, and for $N\in\{6,8,12\}$ in~\cite{gurak2004period}, with quite complex expression.
\end{remark}

\begin{example}
Let $(p,m)=(3,4)$ and $N=4$. Then $r=81$ and $n=20$.
Suppose that $\alpha$ is a primitive element of $\mathbb{F}_{81}$.\\

(1) For the case of $D=C_1^{(4,r)}=\alpha\langle\alpha^4\rangle$.

By Theorem \ref{lem:cwe order 41}, the code $C_{D}$ of \eqref{def:CD} is a $[20,4,12]$ cyclic code over $\mathbb{F}_3$ with complete weight enumerator
\begin{eqnarray*}
CWE(C_{D})=w_0^{20}+60w_0^{8}w_1^{6}w_2^{6}+20w_0^{2}w_1^{9}w_2^{9},
\end{eqnarray*}
which is consistent with numerical computation by Magma.

This code is the best ternary cyclic code of length $20$ and dimension $4$ due to the tables given by Ding~\cite{ding2014differencesets}.

(2) For the case of $D=C_0^{(4,r)}\bigcup C_1^{(4,r)}= \langle\alpha^4\rangle\bigcup \alpha\langle\alpha^4\rangle$.

By Theorem \ref{lem:cwe order 42}, the code $C_{D}$ of \eqref{def:CD} is a $[40,4,24]$ linear code over $\mathbb{F}_3$ with complete weight enumerator
\begin{eqnarray*}
CWE(C_{D})=w_0^{40}+40w_0^{16}w_1^{12}w_2^{12}+40w_0^{10}w_1^{15}w_2^{15},
\end{eqnarray*}
which is consistent with numerical computation by Magma.

(3) For the case of $D=C_0^{(4,r)}\bigcup C_1^{(4,r)}\bigcup C_2^{(4,r)}$.

By Theorem \ref{lem:cwe order 43}, the code $C_{D}$ of \eqref{def:CD} is a $[60,4,36]$ linear code over $\mathbb{F}_3$ with complete weight enumerator
\begin{eqnarray*}
CWE(C_{D})=w_0^{60}+20w_0^{24}w_1^{18}w_2^{18}+60w_0^{18}w_1^{21}w_2^{21},
\end{eqnarray*}
which is consistent with numerical computation by Magma.\\

\end{example}

\begin{example}
Let $(p,m)=(5,4)$ and $N=4$. Then $r=625$ and $n=156$.
Suppose that $\alpha$ is a primitive element of $\mathbb{F}_{625}$.\\

(1) For the case of $D=C_1^{(4,r)}=\alpha\langle\alpha^4\rangle$.

By Theorem \ref{lem:cwe order 41}, the code $C_{D}$ of \eqref{def:CD} is a $[156,4,112]$ cyclic code over $\mathbb{F}_5$ with complete weight enumerator
\begin{eqnarray*}
CWE(C_{D})&=&w_0^{156}+156w_0^{44}(w_1 w_2 w_3 w_4)^{28}+156w_0^{32}(w_1 w_2 w_3 w_4)^{31}\\
&&+156w_0^{28}(w_1 w_2 w_3 w_4)^{32}+156w_0^{20}(w_1 w_2 w_3 w_4)^{34},
\end{eqnarray*}
which is consistent with numerical computation by Magma.

(2) For the case of $D=C_1^{(4,r)}\bigcup C_3^{(4,r)}= \alpha\langle\alpha^4\rangle\bigcup \alpha^3\langle\alpha^4\rangle$.

By Theorem \ref{lem:cwe order 42}, the code $C_{D}$ of \eqref{def:CD} is a $[312,4,240]$ linear code over $\mathbb{F}_5$ with complete weight enumerator
\begin{eqnarray*}
CWE(C_{D})=w_0^{312}+312w_0^{72}(w_1 w_2 w_3 w_4)^{60}+312w_0^{52}(w_1 w_2 w_3 w_4)^{65},
\end{eqnarray*}
which is consistent with numerical computation by Magma.\\

(3) For the case of $D=C_0^{(4,r)}\bigcup C_3^{(4,r)}= \langle\alpha^4\rangle\bigcup \alpha^3\langle\alpha^4\rangle$.

By Theorem \ref{lem:cwe order 42}, the code $C_{D}$ of \eqref{def:CD} is a $[312,4,236]$ linear code over $\mathbb{F}_5$ with complete weight enumerator
\begin{eqnarray*}
CWE(C_{D})&=&w_0^{312}+156w_0^{76}(w_1 w_2 w_3 w_4)^{59}+156w_0^{64}(w_1 w_2 w_3 w_4)^{62}\\
&&+156w_0^{60}(w_1 w_2 w_3 w_4)^{63}+156w_0^{48}(w_1 w_2 w_3 w_4)^{66},
\end{eqnarray*}
which is consistent with numerical computation by Magma.\\

(4) For the case of $D=C_0^{(4,r)}\bigcup C_1^{(4,r)}\bigcup C_2^{(4,r)}$.

By Theorem \ref{lem:cwe order 43}, the code $C_{D}$ of \eqref{def:CD} is a $[468,4,364]$ linear code over $\mathbb{F}_5$ with complete weight enumerator
\begin{eqnarray*}
CWE(C_{D})&=&w_0^{468}+156w_0^{104}(w_1 w_2 w_3 w_4)^{91}+156w_0^{96}(w_1 w_2 w_3 w_4)^{93}\\
&&+156w_0^{92}(w_1 w_2 w_3 w_4)^{94}+156w_0^{80}(w_1 w_2 w_3 w_4)^{97},
\end{eqnarray*}
which is consistent with numerical computation by Magma.\\

\end{example}

\section{Concluding remarks}\label{sec:conclusion}

In this paper, we proposed the complete weight enumerator of a family of
linear code $C_D$ with defining set $D$ constructed from cyclotomy. The formulae for the general strategy and two special cases of $N=3$ and $N=4$ were presented by employing Gaussian periods.
This indicates that the complete weight enumerator of $C_D$ can be determined by the explicit Gaussian periods. As is well known that, the determination of Gaussian periods is quite complicated, so is the complete weight enumerator of $C_D$.

\ifCLASSOPTIONcaptionsoff
  \newpage
\fi

\end{document}